\documentclass[letterpaper, 10 pt, conference]{ieeeconf}

\IEEEoverridecommandlockouts
\usepackage{graphicx}
\graphicspath{{figs/}}
\usepackage{amsmath,amssymb}
\usepackage{tikz-cd}
\usepackage{stmaryrd}
\usepackage{xcolor}
\usepackage{csquotes}
\usepackage{tikz}
\usepackage{caption}
\usepackage{subcaption}
\usepackage[ruled,linesnumbered]{algorithm2e}
\usepackage{xcolor}
\usepackage{booktabs}
\usepackage{microtype}
\usepackage{cite}

\newtheorem{theorem}{Theorem}
\newtheorem{corollary}{Corollary}
\newtheorem{lemma}{Lemma}
\newtheorem{proposition}{Proposition}
\newtheorem{problem}{Problem}
\newtheorem{definition}{Definition}

\newtheorem{assumption}{Assumption}

\newcommand{\define}[1]{\textit{#1}}
\newcommand{\join}{\vee}
\newcommand{\meet}{\wedge}
\newcommand{\bigjoin}{\bigvee}
\newcommand{\bigmeet}{\bigwedge}
\newcommand{\jointimes}{\boxplus}
\newcommand{\meettimes}{\boxplus'}
\newcommand{\bigjoinplus}{\bigjoin}
\newcommand{\bigmeetplus}{\bigmeet}
\newcommand{\joinplus}{\join}
\newcommand{\meetplus}{\meet}

\newcommand{\semimod}{\mathcal{S}}
\newcommand{\graph}{\mathcal{G}}
\newcommand{\nodes}{\mathcal{V}}

\newcommand{\edges}{\mathcal{E}}
\newcommand{\neighbors}{\mathcal{N}}
\newcommand{\Weights}{\mathcal{A}}
\renewcommand{\leq}{\leqslant}
\renewcommand{\geq}{\geqslant}
\renewcommand{\preceq}{\preccurlyeq}
\renewcommand{\succeq}{\succcurlyeq}
\newcommand{\Rmax}{\mathbb{R}_{\mathrm{max}}}
\newcommand{\Rmin}{\mathbb{R}_{\mathrm{min}}}

\newcommand{\R}{\mathbb{R}}

\newcommand{\A}{\mathbf{A}}
\newcommand{\B}{\mathbf{B}}
\newcommand{\x}{\mathbf{x}}

\newcommand{\X}{\mathbf{X}}
\newcommand{\W}{\mathbf{W}}
\newcommand{\weights}{\mathcal{W}}
\newcommand{\alternatives}{\mathcal{X}}

\newcommand{\y}{\mathbf{y}}
\newcommand{\Y}{\mathbf{Y}}
\newcommand{\z}{\mathbf{z}}
\newcommand{\Z}{\mathbf{Z}}

\renewcommand{\b}{\mathbf{b}}
\newcommand{\I}{\mathbf{I}}

\newcommand{\Laplacian}{\mathcal{L}}
\newcommand{\F}{\mathcal{F}}
\newcommand{\inv}[1]{{#1}^{\sharp}}

\begin{document}

\title{\bf Max-Plus Synchronization in Decentralized Trading Systems}

\author{Hans Riess$^{1\ast}$, Michael Munger$^{2,3}$, and Michael M. Zavlanos$^{1,4}$
\thanks{$^{\ast}$Corresponding author's email: \textbf{}\texttt{hans.riess@duke.edu}.}%
\thanks{$^{1}$Department of Electrical and Computer Engineering, Duke University.}%
\thanks{$^{2}$Department of Political Science, Duke University.}%
\thanks{$^{3}$Department of Economics, Duke University.}%
\thanks{$^{4}$Dept.~of Mechanical Engineering \& Materials Science, Duke University.}%
\thanks{This work is supported in part by ONR under agreement \#N00014-18-1-2374 and by AFOSR under the award \#FA9550-19-1-0169.}%
}

\maketitle

\begin{abstract}
We introduce a decentralized mechanism for pricing and exchanging alternatives constrained by transaction costs. We characterize the time-invariant solutions of a heat equation involving a (weighted) Tarski Laplacian operator, defined for max-plus matrix-weighted graphs, as approximate equilibria of the trading system. We study algebraic properties of the solution sets as well as convergence behavior of the dynamical system. We apply these tools to the ``economic problem'' of allocating scarce resources among competing uses. Our theory suggests differences in competitive equilibrium, bargaining, or cost-benefit analysis, depending on the context, are largely due to differences in the way that transaction costs are incorporated into the decision-making process. We present numerical simulations of the synchronization algorithm (RRAggU), demonstrating our theoretical findings.

\end{abstract}

\section{Introduction}
\label{sec:intro}
Max-plus algebra, originally ``minimax algebra'' \cite{cunninghamegreen1979}, is an algebraic theory derived from ordinary linear algebra by substituting the addition operation with maximum and the multiplication operation with ordinary addition. Max-plus algebra is a special case of tropical geometry \cite{maclagan2021}, a relatively new mathematical field concerned with simplifying difficult problems through \emph{tropicalization}---converting everything to max-plus or min-plus arithmetic. The historical development of max-plus algebra has been largely motivated by applications in discrete event systems (DESs), a class of dynamical systems characterized by a discrete state space and event-driven state transactions \cite{cassandras2008}.
The application we propose here, multi-agent economic exchange, lies outside the scope of DESs and is more closely related to efforts to analyze and control markov decision processes (MDPs) with max-plus algebra \cite{berthier2020,chandrashekar2014,goncalves2021,akian2008}: states are value-vectors and state transitions are constrained by transaction costs (e.g.~rewards or penalties to transact).

The “economic problem” is usually characterized as the optimal allocation of scarce resources among potentially competing uses \cite{robbins2007}. Resources are scattered and information about their quality and value is “dispersed” \cite{hayek1945}. Worse, information about value is not just dispersed, but actually unknown, since the “transaction costs” and stated values of bidders and sellers of a potential exchange emerge from the process of negotiation  itself \cite{coase1960}. In economic systems, transaction costs---locating, negotiating terms, packaging, and delivering---play a role analogous to friction in physical systems: energy (value) is lost in the form of heat rather than being available for use.


We suggest that the problem of transaction costs as an impediment to identifying and implementing otherwise valuable exchanges is in fact generic, and suggest that max-plus is a useful means of unifying the underlying structure of the three apparently unrelated, but, in fact, closely connected, aggregation mechanisms in economics: “competitive equilibrium” of market processes \cite{debreu1959,arrow1971}, "bargaining" \cite{coase1960}, and “cost-benefit analysis” \cite{hicks1939,kaldor1939}.

The tropical Tarski Laplacian (Definition \ref{def:tarski}) driving the dynamics of our decentralized trading mechanism closely resembles the alternating method  for solving the two-sided equation $\A \jointimes \x = \B \jointimes \y$ \cite{cunninhamegreene2003}. The motivation for the alternating method, shared in our work, is synchronization of coupled max-plus systems. The key difference for our application is the identification of the evolution of the system toward a convergence, defined as a set of allocations where all resources have “found” their highest valued uses, up to divergences caused by transaction costs of discovering or implementing further exchanges.
Beyond generalizing the alternating method from a single coupling to an arbitrary (undirected) graph of couplings, we note that the updates of the heat equation \eqref{eq:heat} induced by the Tarski Laplacian are synchronous, in contrast to the alternating method, whose updates alternate between two agents.

\subsection{Related Work}
\label{sec:related}

Recently, there has been a resurgence in activity in max-plus algebra focused on applications in control theory \cite{maragos2017}, signal processing \cite{theodosis2018} (including graph signal processing \cite{blusseau2018}), and machine learning \cite{maragos2021} (including deep learning \cite{alfarra2022,zhang2019,tsilivis2022}). Current research in applied tropical geometry extends to statistics \cite{pachter2004} and even optimal transport \cite{lee2022}. A recent focus has been on finding the sparsest possible solutions to one-sided max-plus linear systems \cite{tsiamis2019,tsilivis2022}, leading to natural applications in discrete event systems \cite{tsiamis2019}, optimal control \cite{blusseau2018}, and multivariate convex regression \cite{tsilivis2021}.

The present work also connects to a recent model of multidimensional opinions dynamics \cite{hasnen2021}, more broadly construed as ``social information dynamics'' \cite{riess2022,ghirst2022} drawing on the recently introduced theory of sheaf Laplacians \cite{hansen2019}. Our decentralized method for arriving at (approximate) equilibrium is equivalent to max consensus \cite{nejad2009} in the one-dimensional case ($d =1$). Invoking the insights that sheaf theory \cite{curry2014} and sheaf Laplacians \cite{hansen2019} bring to generalized consensus/coherence problems, the Tarski Laplacian was introduced in its full abstract form \cite{riess2022b} as an operator of assignments of lattice-valued data over an undirected graph tethered together by Galois connections: lattice-valued sheaves \cite{ghrist2022}. Since then, another instance of the Tarski Laplacian was introduced to solve multi-agent knowledge consensus problems \cite{riess2022}. Convergence guarantees in that work \cite{riess2022} rely on the lattices satisfying the descending chain condition \cite{roman2008}. In the present work, we do not have this guarantee, posing a new challenge to convergence analysis.


\subsection{Outline}
\label{sec:outline}

In Section \ref{sec:background} we review background material in max-plus algebra. Then, in Section \ref{sec:problem}, we formulate the synchronization (Problem \ref{prob:exact}) and approximate synchronization problem under consideration (Problem \ref{prob:approximate}). In Section \ref{sec:solution}, we propose an algorithm (Algorithm \ref{alg:sync}) based on the tropical Tarski Laplacian (Definition \ref{def:tarski}) to solve the approximate problem and, in Section \ref{sec:solutions}, we analyze the convergence of the algorithm as well as the properties of the converging solutions. Finally, in Section \ref{sec:experiments}, we present numerical experiments to illustrate the performance and scalability of our proposed algorithm.

\section{Background}
\label{sec:background}

Let $\Rmax$ denote the \define{max-plus} (tropical) semiring $\R \cup \{-\infty\}$ with the operations
\begin{align*}
    \alpha \joinplus \beta &\triangleq& \max\{\alpha, \beta\}, \\
    \alpha + \beta &\triangleq& \alpha + \beta.
\end{align*}
This semiring has (additive) unit $-\infty$ and (multiplicative) unit $0$. Similarly, let $\Rmin$ denote the \define{min-plus} (tropical) semiring $\R \cup \{-\infty\}$ with operations
\begin{align*}
    \alpha \meetplus \beta &\triangleq& \min\{\alpha, \beta\}, \\
    \alpha +' \beta &\triangleq& \alpha + \beta,
\end{align*}
and (additive) unit $\infty$ and (multiplicative) unit $0$.

Both of these semirings define an alternative arithmetic, leading to max-plus and min-plus linear algebra, respectively.
Suppose $\A \in \Rmax^{m \times p}$ and $\B \in \Rmax^{p \times n}$. Then $\A \jointimes \B$ is an $m$-by-$n$ matrix defined as
\begin{align}
    \left[ \A \jointimes \B \right]_{i,j} &=& \bigjoinplus_{k=1}^{p} [\A]_{i,k} + [\B]_{k,j}. \label{eq:matrix-multiply}
\end{align}
The identity matrix $\I \in \Rmax^{n \times n}$ is defined $[\I]_{i,j} = 0$ if $i=j$, $[\I]_{i,j} = -\infty$, otherwise.
Similarly, if $\A \in \Rmin^{m \times p}$ and $\B \in \Rmin^{p \times n}$, then
\begin{align}
    \left[ \A \meettimes \B \right]_{i,j} &=& \bigmeetplus_{k=1}^{p} [\A]_{i,k} +' [\B]_{k,j}. \label{eq:matrix-multiply-2}
\end{align} 
The \define{pseudoinverse} of $\A \in \Rmax^{m \times n}$ is the matrix $\inv{\A} \in \Rmin^{n \times m}$ defined as $[\inv{\A}]_{i,j} = -[\A]_{j,i}$.

The set $\Rmax^n$ with operations analogous to vector addition and scalar multiplication, i.e.,
\begin{align*}
    [\x \join \y]_i &=& \max \{ [\x]_i, [\y]_i \}, \\
    [\alpha + \x]_i &=& \alpha +[\x]_i,
\end{align*}
for $\x, \y \in \Rmax^n,~\alpha \in \Rmax$, is an example of a \define{semimodule}.
Semimodules are the analogues of vector spaces in both the max- (min-) plus setting. Suppose $\{\x_k\} \in \Rmax^n$ and $\{\alpha_k\} \in \Rmax$, then their \define{max-plus linear combination} is the vector
\begin{align*}
    \y = \bigjoin_{k=1}^{K} \left(\alpha_i + \x_{k}\right).
\end{align*}
A subset of $\Rmax^n$ closed under max-plus linear combinations is called a \define{subsemimodule}, and is of special interest. 

If $\A \in \Rmax^{m \times n}$ is a matrix and $\x \in \Rmax^n$ is a column vector, then the multiplication
\begin{align*}
    \left[ \A \jointimes \x \right]_i &=& \bigjoin_{j=1}^{n} [\A]_{i,j} + [\x]_j
\end{align*}
defines a transformation between the semimodules $\Rmax^n$ and $\Rmax^m$. The following assumption is standard.

\begin{definition}[\cite{cunninghamegreen1979}]
    A matrix $\A \in \Rmax^{m \times n}$ is \define{doubly G-astic} if every row and column of $\A$ has at least one entry greater than $-\infty$.
\end{definition}

Max-plus transformations have the following properties (analogous properties hold for min-plus transformations).
\begin{lemma}[{\cite{butkovic2010}}] \label{lem:nonlinear}
Suppose $\A \in \Rmax^{m\times n}$, $\x, \y \in \Rmax^n$, and $\alpha \in \Rmax$. Then,
\leavevmode
\begin{enumerate}
    \item $\A \jointimes \left( \x \join \y \right) = \left(\A \jointimes \x \right) \join \left( \A \jointimes \y\right)$.
    \item $\A \jointimes \left(\alpha + \x \right) = \alpha + \A \jointimes \x$.
    \item If $\x \preceq \y$, then $\A \jointimes \x \preceq \A \jointimes \y$.
\end{enumerate}
\end{lemma}

The semimodule $\Rmax^n$ can be viewed as a partially ordered set under the product order: $\x \preceq \y$ if and only if $x_i \leq x_i,~\forall i \in \{1,2,\dots, n\}$
with supremum and infimum operators
\begin{align}
    [\x \join \y]_i &=& \max \{ [\x]_i, [\y]_i \} \label{eq:join} \\
    [\x \meet \y]_i &=& \min \{ [\x]_i, [\y]_i \} \label{eq:meet}
\end{align}
called \define{join} and \define{meet}, respectively, making $\Rmax^n$ a \define{lattice} \cite{roman2008}. Semimodules, such as $\Rmax^n$ or $\Rmin^n$, with the additional structure of a lattice have been called weighted lattices in the literature \cite{maragos2017,tsilivis2022}. The following result follows from \emph{residuation theory} \cite{blyth2014}, closely related to Galois connections \cite[Chapter 7]{davey2002}.
\begin{lemma}[{\cite[Lemma 7.26]{davey2002}}] \label{lem:adjoint}
    Suppose $\A \in \R^{m \times n}$. Then, for all $\x \in \Rmax^n$, $\y \in \Rmin^m$
    \begin{enumerate}
        \item $\A \jointimes \x \preceq \y$ if and only if $\x \preceq \inv{\A} \meettimes \y$
        \item $\inv{\A} \meettimes \left( \A \jointimes \x \right) \succeq \x$ and $\A \jointimes \left( \inv{\A} \meettimes \y \right) \preceq \y$
    \end{enumerate}
\end{lemma}

One consequence of Lemma \ref{lem:adjoint} is that it characterizes (sub)solutions of max-plus matrix equations. Suppose $\A \in \Rmax^{m \times n}$ and $\b \in \Rmax^n$. Then, the residuation $\bar{\x} = \inv{\A} \meettimes \b$, called the \define{principal solution}, is the greatest solution of $\A \jointimes \x = \b$, if a solution exists, otherwise, the greatest subsolution, i.e.~a vector $\x$ such that $\A \jointimes \x \preceq \b$ \cite[Proposition 1.1]{cunninghamegreen1979}. The vector $\inv{\A} \meettimes \b$ is \define{finite}, i.e.~having no entries equal to $-\infty$, if $\b$ is finite and $\A$ is doubly G-astic. However, non-finite subsolutions (less than $\bar{\x}$ in the product order) are shown to be found by selecting entries of $\bar{\x}$ to be $-\infty$ in a greedy algorithm \cite{tsiamis2019}.


\section{Problem Definition}
\label{sec:problem}

Consider a multi-agent trading system, where pairs of \define{agents} compare the value of \define{alternatives} often leading to an exchange.
We assume agents are collected in a finite set $\{1,2,\dots, N\}$ and alternatives are collected in a finite set $\alternatives = \{1,2,\dots,d\}$. We use letters $u,v,w$ to denote agents and letters $i,j,k$ to denote alternatives.

(Reservation) values detail the (minimal) worth of each alternative to each agent. Agents value alternatives to various degrees based on individual preferences and supply. We represent the \define{reservation value} of alternatives by Agent $u$ with a real-valued function on domain $\alternatives$, equivalently, and conveniently, written as a column vector $\X_u \in \Rmax^d$ that collects the values of all alternatives to Agent $u$. Value vectors may change over time depending on interactions with other agents such as negotiation or bargaining. We assume values are normalized with respect to a standardized currency. If an agent doesn't include Alternative $i$ in negotiations, we set $[\X_u]_i = -\infty$.

Agents can propose transactions with other agents. An agent that proposes the transaction is called a \define{seller} and an agent who may accept a transaction is called a \define{bidder}.  We assume all agents may assume the role of a bidder or a seller and may openly establish trade relationships with other agents. Transactions, therefore, take place in a decentralized fashion. A \define{transaction} consists of the exchange of an Alternative $i$ for another Alternative $j$ less some cost which we call a \define{transaction cost}. We assume transaction costs are dependent on: the alternative to be traded (Alternative $i$), the alternative to be acquired (Alternative $j$), the identity of the seller (Agent $u$), and the identity of the bidder (Agent $v$). We model transaction costs via max-plus matrices. Specifically, for the seller to exchange Alternative $j$ for Alternative $i$ with the bidder, let $\A_{u,v} \in \Rmax^{d \times d}$ be a matrix with entries $\left[ \A_{u,v} \right]_{i,j}$ representing a transaction value (negative transaction cost) of Agent $u$ trading Alternative $j$ for Alternative $i$ with Agent $v$ (see Table \ref{table:transaction-costs}).

\begin{table}
\centering
\begin{tabular}{cl}
\toprule
$[\A_{u,v}]_{i,j}$ & Value for exchanging $j$ for $i$ \\
\midrule
$-\infty$ & exchange cannot be made \\
$< 0$ & transaction cost for exchange \\
$0$ & no transaction cost \\
$> 0$ & transaction value (e.g.~subsidy) for exchange \\
\bottomrule
\end{tabular}
\caption{}
\label{table:transaction-costs}
\end{table}

We define the \define{effective (supply) value} of Alternative $i$ to the seller relative to negotiations with the bidder as
\begin{align}
    \left[ \A_{u,v} \jointimes \X_u \right]_i = \max_{j=1}^{d} \Bigl\{[\X_u]_j + [\A_{u,v}]_{i,j} \Bigr\} \label{eq:relative-demand}
\end{align}
where $[\X_u]_i$ is the reservation value of Alternative $i$.
For instance, if Agent $u$ was to participate in an exchange with Agent $v$, the quantity \eqref{eq:relative-demand} describes the realized value of the exchange(s) with the highest value less transaction costs. In this paper, we are interested in the equilibrium condition
\begin{align}
    \max_{j=1}^{d} \{ [\X_u]_j + [\A_{u,v}]_{i,j} \} = \max_{j=1}^{d} \{ [\X_v]_j + [\A_{v,u}]_{i,j} \} \forall i \in \alternatives, \nonumber
\end{align}
which we call the \define{value equation}. In matrix form, the value equation can be rewritten as
\begin{align}
        \A_{u,v} \jointimes \X_{u} = \A_{v,u} \jointimes \X_v .\label{eq:demand-matrix}
\end{align}


As noted in the introduction, the intuition behind the equilibrium condition is deceptively simple---it is not possible to make any additional voluntary exchanges, because every feasible exchange would make at least one participant worse off.

\begin{problem}[Effective Value Equilibrium] \label{prob:exact}
    In a decentralized trading system, determine a value of every alternative for every agent such that the effective value equation \eqref{eq:demand-matrix} is satisfied for every alternative and every pair of trading partners.
\end{problem}

The equilibrium condition \eqref{eq:demand-matrix}, mathematically, is always feasible, i.e.~$\x = \y = [-\infty~-\infty \cdots -\infty]^{\top}$ is a solution to the two sided equation $\A \jointimes \x = \B \jointimes \y$. However, finite solutions are not guaranteed \cite{butkovic2006}. Thus, we relax the value equation \eqref{eq:demand-matrix} to
\begin{align}
    \| \A_{u,v} \jointimes \X_{u} - \A_{v,u} \jointimes \X_v \|_{\infty} \leq \epsilon . \label{eq:approx-demand}
\end{align}

\begin{problem}[Approximate Value Equilibrium]\label{prob:approximate}
    In a decentralized trading system, determine a value of every alternative for every agent such that the approximate value equation \eqref{eq:approx-demand} is satisfied for every pair of trading partners.
\end{problem}

\section{Decentralized Solution}
\label{sec:solution}

We model economic relationships between agents by an undirected graph $\graph = (\nodes,\edges)$. The edges in this trade network are pairs of agents $(u,v) \in \edges$ such that Agent $u$ can propose a trade with Agent $v$. We assume that the possibility of trade between two agents is a symmetric relationship. We also define the (graph) neighborhood $\neighbors_u = \{v \in \nodes~\vert~(u,v) \in \edges \}$ of Agent $u \in \nodes$ as the set of all possible trading partners of Agent $u$.

We introduce two (edge) weightings on $\graph$, a scalar weighting and a matrix weighting: $\weights : \nodes \times \nodes \to \Rmin$ so that $(u,v) \mapsto [\W]_{u,v}$, where
$\W \in \Rmin^{N \times N}$, and $\Weights : \nodes \times \nodes \to \Rmax^{d \times d}$ so that $(u,v) \mapsto \A_{u,v}$. We make the following assumptions concerning weights.

\begin{assumption} \label{assume:scalar}
    The matrix $\W$ is symmetric, and has the sparsity pattern of $\graph = (\nodes,\edges)$, i.e.~$[\W]_{u,v} < \infty$ if $(u,v) \in \edges$, otherwise  $[\W]_{u,v}=\infty$.
\end{assumption}


\begin{assumption} \label{assume:matrix}
    For all $(u,v) \in \edges$, the matrices $\A_{u,v} \in \Rmax^{d \times d}$ and $\A_{v,u} \in \Rmax^{d \times d}$ are doubly G-astic.
\end{assumption}

\subsection{The tropical Tarski Laplacian}
\label{sec:tarski}

Given the data $\graph = (\nodes,\edges,\weights,\Weights)$, we introduce an operator on global value vectors which we call the \define{tropical Tarski Laplacian}. Without the scalar weighting term, the definition of the tropical Tarski Laplacian below is a specialization of the Tarski Laplacian defined in previous works \cite{riess2022,ghrist2022,riess2022b}. 

\begin{definition}[Tropical Tarski Laplacian] \label{def:tarski}
    Suppose $\graph = (\nodes,\edges,\weights,\Weights)$, and suppose $\X \in \left(\Rmax^d \right)^N$.
    The \define{tropical Tarski Laplacian} is an operator $\Laplacian : \left(\Rmax^{d}\right)^N \longrightarrow \left(\Rmax^{d}\right)^N$ defined block-wise
    \begin{align}
        \Laplacian(\X)_u = \bigmeetplus_{v \in \neighbors_u} [\W]_{u,v} +' \inv{\A}_{u,v} \meettimes \left( \A_{v,u} \jointimes \X_v \right). \label{eq:tarski}
    \end{align}
\end{definition}


\begin{proposition} \label{prop:opt}
Suppose Agent $u$ has neighbors with value-vectors $\{\X_v \}_{v \in \neighbors_u}$. Then, $\Laplacian(\X)_u$ is the optimal solution to the following bi-level optimization problem
   \begin{align*}
    & \bigmeet_{v \in \neighbors_u} [\W]_{u,v} +' \Z_v \\
    & \text{subject to} \\
    & \Z_v \in \mathrm{argmax} \{ \Y \in \Rmax^d~\vert~\A_{u,v} \jointimes \Y \preceq \A_{v,u} \jointimes \X_v \}.
    \end{align*}
\end{proposition}
The tropical Tarski Laplacian, thus, is a aggregation mechanism for values trading partners assign to alternatives.  

\subsection{The Heat Equation}
\label{sec:heat}

Given the data $\graph = (\nodes,\edges,\weights,\Weights)$ and initial condition $\X(0) \in \left(\Rmax^{d}\right)^N$, we define a discrete-time time-invariant dynamical system called the \define{heat equation} as
\begin{align} 
    \X(t+1) &=& \Laplacian \left( \X(t) \right) \meet \X(t). \label{eq:heat}
\end{align}

Locally, trajectories of the heat equation consist of the following steps (Algorithm \ref{alg:sync}): a bidder communicates their effective supply values to a seller who performs a residuation (Line 6), the seller re-scales the effective supply values by $[\W]_{u,v}$ (Line 7), the seller aggregates the resulting values by computing a meet (Line 10), and the seller updates her supply value by computing the meet of the resulting value and her prior value, $\X_u$ (Line 12).
Finally, we impose a stopping condition (Line 2) 
\begin{align}
        \| \A_{u,v} \jointimes \X_u - \A_{v,u} \jointimes \X_v \|_{\infty} \leq [\W]_{u,v} \quad \forall (u,v) \in \edges,
        \nonumber
    \end{align}
which we enforce using the following loss function
\begin{align}
    \ell(\X) &=& \max_{(u,v) \in \edges} \| \A_{v,u} \jointimes \X_v - \A_{u,v} \jointimes \X_u \|_{\infty}. \label{eq:loss}
\end{align}

In Section \ref{sec:solution}, we analyze the convergence of Algorithm \ref{alg:sync}. While convergence is not guaranteed \emph{a priori}, provided the algorithm terminates, the algorithm returns a solution to Problem \ref{prob:approximate} because $\ell(\X) \leq \epsilon$ if and only if
\eqref{eq:approx-demand} is satisfied.

\begin{algorithm}[t]
\caption{ RRAggU} \label{alg:sync}
\SetKwComment{Comment}{/* }{ */}
\KwData{$\graph = (\nodes,\edges,\Weights,\weights)$; $\epsilon > 0$; $\X \in \Rmax^{N \cdot d}$}
\KwResult{$\X \in \Rmax^{N \cdot d}$}
$\mathsf{loss} = \infty$ \\
\While{$\mathsf{loss} > \epsilon$}{
    \ForPar{$u \in \nodes$}{
    $\z \gets [\infty~\infty~\cdots~\infty]^{\top}$ \\
    \For{$v \in \neighbors_u$}{
        $\mathsf{Residuate}(u,v) \gets \inv{\A}_{u,v} \meettimes (\A_{v,u} \jointimes \X_v) $\\
        $\mathsf{Rescale}(u,v) \gets [\W]_{u,v} +' \mathsf{Residuate}(u,v)$ \\
        $\z \gets \z \meet \mathsf{Rescale}(u,v)$
    } 
    $\mathsf{Aggregate}(\neighbors_u) \gets \z$ \\
    $\mathsf{Update}(u) \gets \X_u \meet \mathsf{Aggregate}(\neighbors_u)$ \\
    $\X_u \gets \mathsf{Update}(u)$
    }
    $\mathsf{loss} \gets \ell(\X)$
    }
\end{algorithm}

\section{Convergence Analysis and Equilibria}
\label{sec:solutions}

\noindent In this section, we
\begin{enumerate}
    \item Analyze the convergence of Algorithm \ref{alg:sync}; \\
    \item Supply an algebraic characterization of the time-invariant solutions of the heat equation,
        \begin{align}
            \semimod = \{ \X \in \left(\Rmax^{d}\right)^N~\vert~\Laplacian(\X) \meet \X = \X \};\label{eq:solutions}
        \end{align}
    \item Show that the solutions of the heat equation are solutions of the global approximate value equation
\begin{align}
\| \A_{u,v} \jointimes \X_{u} - \A_{v,u} \jointimes \X_v \|_{\infty} \leq \epsilon \quad \forall (u,v) \in \edges. \label{eq:equilibrium-global}
\end{align}
\end{enumerate}

\subsection{Convergence of the heat equation}
\label{sec:convergence}
We show the updates of the heat equation exhibit contracting behavior, in the limit. 
For this analysis, we review a few standard definitions \cite{gunawardena1995,cochet1999}.
A map $\F: \R^n \to \R^n$ is
    \begin{enumerate}
    \item \define{Monotonic} if $\x \preceq \y$ implies $\F(\x) \preceq \F(\y)$,
    \item Additively \define{homogeneous} if $\F(\x + \alpha) = \F(\x) + \alpha$ for all $\x \in \R^n$, $\alpha \in \R$,
    \item \define{Non-expansive} if $\|\F(\x) - \F(\y)\|_{\infty} \preceq \| \x - \y \|_{\infty}$ for all $\x, \y \in \R^n$.
\end{enumerate}

For the following, assume $\graph = \bigl(\nodes,\edges,\weights,\Weights \bigr)$ satisfies Assumptions \ref{assume:scalar} and \ref{assume:matrix}, and $\X_u(0)$ is finite for every $u \in \{1,2,\dots, N\}$.

\begin{proposition} \label{prop:topical}
     Suppose $\F: \left(\R^{d}\right)^N \longrightarrow \left(\R^{d}\right)^N$ is defined $\F(\X) = \Laplacian(\X) \meet \X$. Then, under the above assumptions, $\F$ is non-expansive.
\end{proposition}
\begin{proof}
    See Appendix.
\end{proof}

Let $\X(t)$ be a trajectory of the heat equation with initial condition $\X(0)$. Then, using Proposition \ref{prop:topical} we can show the following result.
\begin{theorem} \label{thm:convergence}
There exists a scalar $\alpha\geq 0$ so that
\begin{align*}
    \lim_{t \to \infty} \| \X(t) - \X(t+1) \|_{\infty} = \alpha.
\end{align*}
\end{theorem}
\begin{proof}
See Appendix.
\end{proof}

Note that if $\alpha > 0$, then $\X(t)$ does not converge. In this case, Theorem \ref{thm:convergence} says that the $\ell_\infty$-distance between consecutive iterates of the heat equation asymptotically approaches a fixed value $\alpha > 0$, which depends on the data $\graph = \bigl(\nodes,\edges,\weights,\Weights \bigr)$ and $\X(0)$. This case has an interesting economic interpretation, which we discuss in Section \ref{sec:experiments}.
On the other hand, if $\alpha = 0$, then $\| \X(t+1) - \X(t) \|_{\infty} \to 0$ as $t \to \infty$. It follows the value of every alternative determined by each agent converges point-wise, giving rise to the following result, whose proof is trivial.
\begin{corollary} \label{cor:converge}
    Suppose $\| \X(t+1) - \X(t) \|_{\infty} \to 0$ as $t \to \infty$. Then,
    \begin{align*}
        \lim_{t \to \infty} [\X_u(t+1)]_i - [\X_u(t)]_i = 0
    \end{align*}
    for all $u \in \nodes,~i \in \alternatives$.
\end{corollary}%

\subsection{Algebraic structure of equilibria}
\label{sec:semimod}

In this section, we characterize the stable manifold and set of solutions of the heat equation. Specifically, we define the stable manifold as follows.

\begin{definition}[Stable Manifold] \label{def:stable-manifold}
    Given $\alpha\geq0$, and the data $\graph = (\nodes,\edges,\weights,\Weights)$, $\X(0)$, the \define{stable manifold} of the heat equation \eqref{eq:heat} is the following subset: $\mathrm{Stab}_{\alpha}\bigl(\graph\bigr) = \{ \X \in (\R^{d})^N~\vert~ \| \F(\X) - \X \|_{\infty} \leq \alpha \}$.
\end{definition}

As discussed above, if $\alpha = 0$, the heat equation converges to a time-invariant solution.
The following result algebraically characterizes the set $\semimod$ of solutions.

\begin{theorem} \label{thm:semimodule}
    $\semimod$ forms a subsemimodule of $\left(\Rmax^{d}\right)^N$.
\end{theorem}
\begin{proof}
    See Appendix.
\end{proof}

Theorem \ref{thm:semimodule} implies that max-plus linear combinations of solutions remain solutions solutions. As we can always find an $\alpha \in \Rmax$ such that $\X + \alpha$ is positive, we can always produce a positive solution from an arbitrary solution by re-scaling. Theorem \ref{thm:semimodule} also implies, if $\X, \Y \in \semimod$, then the join $\X \join \Y$ reflects a solution with the values of each agent being the alternative-wise maximum of the corresponding values in $\X$ and $\Y$.

\subsection{Effective value equilibrium}
\label{sec:time-invariant}

In this section, we show that the set of solutions $\semimod$ satisfies the effective value equilibrium condition \eqref{eq:approx-demand}. Specifically, suppose $\graph = (\nodes, \edges, \weights, \Weights)$ and $\X \in \semimod$. Then, the $\ell_\infty$-distance between the effective values for each pair $(u,v) \in \edges$ of agents is bounded by $[\W]_{u,v}$. As a corollary, if $\epsilon = \max_{uv \in \edges} [\W]_{u,v}$, then time-invariant solutions of the heat equation are solutions to Problem \ref{prob:approximate}.

\begin{theorem} \label{thm:existence}
    Suppose $\X \in \semimod$. Then,
    \begin{align}
        \| \A_{u,v} \jointimes \X_u - \A_{v,u} \jointimes \X_v \|_{\infty} \leq [\W]_{u,v} \quad \forall (u,v) \in \edges.
        \nonumber
    \end{align}
\end{theorem}
\begin{proof}
    See Appendix.
\end{proof}
\begin{corollary} Suppose $\X \in \semimod$. Then,
\begin{align*}
        \| \A_{u,v} \jointimes \X_u - \A_{v,u} \jointimes \X_v \|_{\infty} \leq \epsilon \quad \forall (u,v) \in \edges.
    \end{align*}
\end{corollary}

\section{Numerical Experiments}
\label{sec:experiments}

We perform several simulations to affirm our theoretical findings as well as visualize the behavior of the heat equation dynamics \eqref{eq:heat}. In the following experiments, we generate a fixed Erd\H{o}s-R\'{e}nyi graph with $N=20$ nodes and probability $p=0.2$ of drawing an edge between nodes $u, v \in \nodes$. Selecting $d=10$ alternatives, we, then, generate (max-plus) matrices $\A_{u,v} \in \R^{d \times d}$ whenever $(u,v) \in \edges$ by selecting entries uniformly at random in $[-1,1]$. Next, we generate (min-plus) scalar edge weights by selecting entries of $\W \in \R^{N \times N}$ uniformly at random in $[0,1]$ if $(u,v) \in \edges$, otherwise, setting $[\W]_{u,v} = \infty$. In a series of $n_{\text{trials}} = 20$ trials, we generate $\X(0) \in \left(\R^{d}\right)^N$ uniformly at random (again from $[-1,1]$) and apply Algorithm \ref{alg:sync} to each initial condition for $t = 0, 1, \dots, 10$ steps (ignoring the stopping condition). For each trajectory, we calculate the loss
\begin{align*}
    \ell\bigl( \X(t) \bigr) &=& \max_{(u,v) \in \edges} \| \A_{u,v} \jointimes \X_u - \A_{v,u} \jointimes \X_v \|_{\infty},
\end{align*}
and the discrete gradient, $\alpha(t)  = \| \X(t+1) - \X(t) \|_{\infty}$.

\begin{figure}[h]
\centering
\begin{subfigure}[b]{0.5\textwidth}
    \centering
    \includegraphics[width=\textwidth]{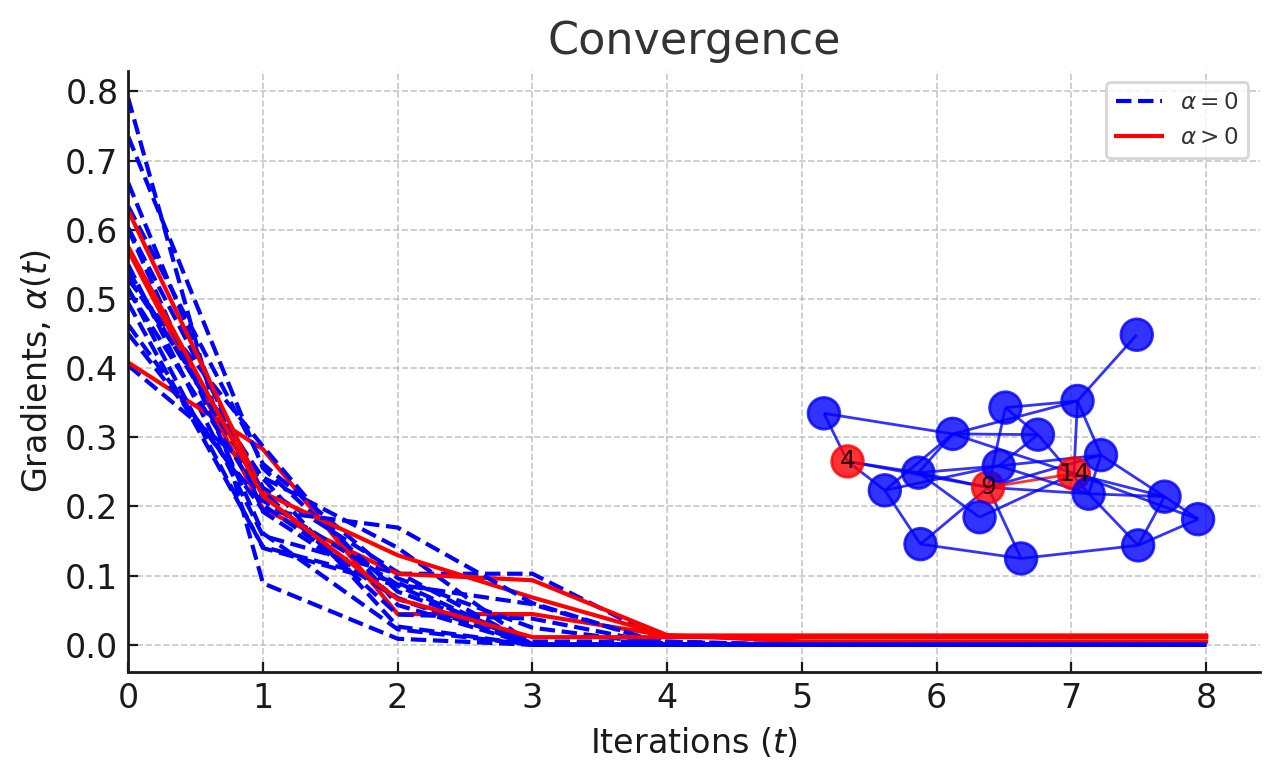}
    \caption{Gradients of trajectories $\alpha(t)$; in sample trajectory ($\alpha>0$), red nodes do not converge.}
    \label{fig:convergence}
\end{subfigure}
\begin{subfigure}[b]{0.5\textwidth}
    \centering
    \hspace{-1cm} \includegraphics[width=\textwidth]{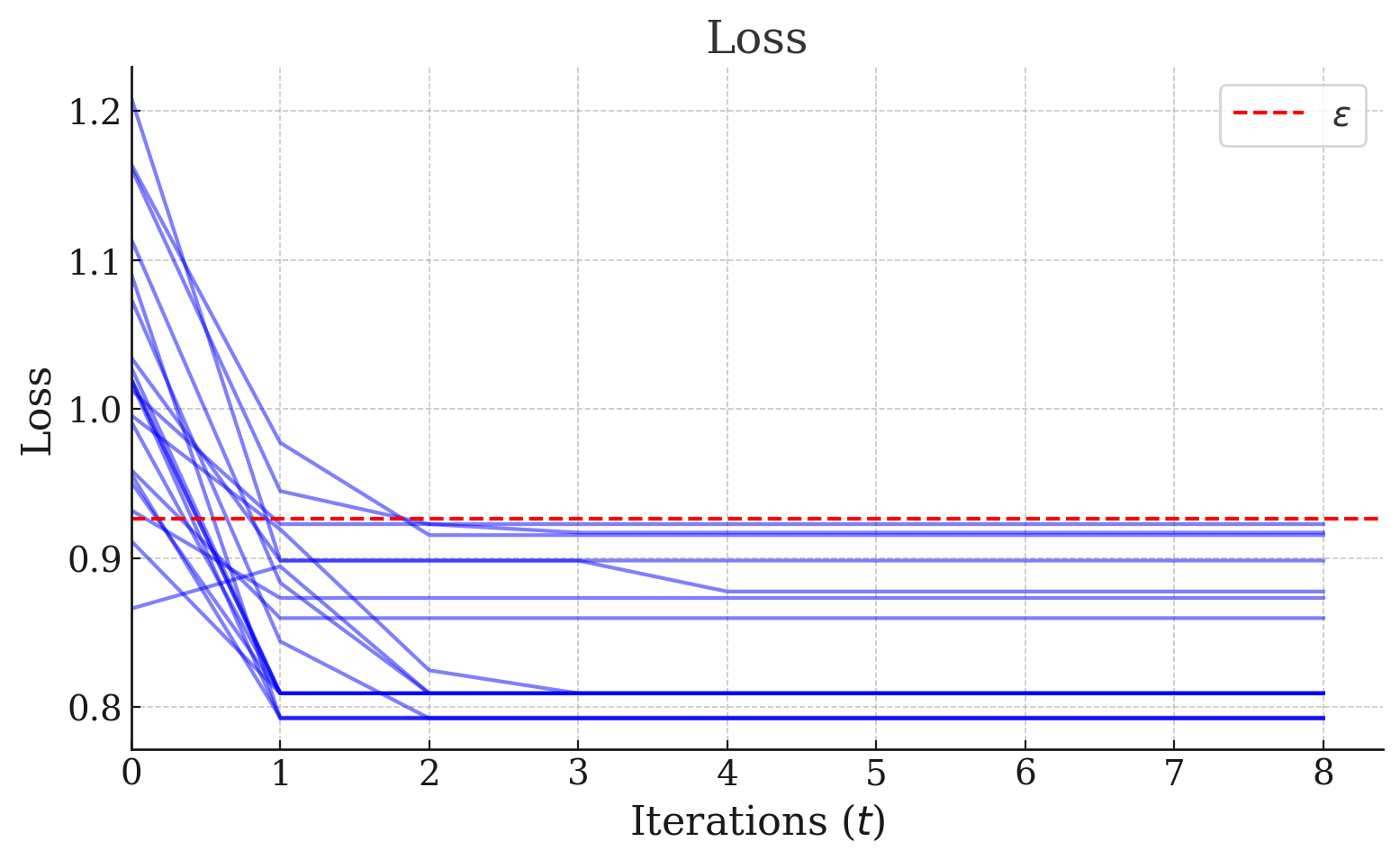}
    \caption{Loss of trajectories $\ell(\X)$.}
    \label{fig:loss}
\end{subfigure}
\caption{Convergence analysis for $n_{\mathrm{trials}} = 20$ initial conditions; $\graph$ is an random graph with $N=20$ nodes; $d=10$ alternatives; $\Weights$ and $\weights$ random.  }
\end{figure}

\subsection{Results}
Fig.~\ref{fig:convergence} shows that $\alpha(t)$ converges to some $\alpha \geq 0$, as expected  from Theorem \ref{thm:convergence}.
Specifically, in all but four of the trials, $\alpha(t)$ converges to $0$, which means that $\X(t)$ also converges with these initial conditions; see Corollary \ref{cor:converge}. For the same set of trials, Fig.~\ref{fig:loss} shows that the loss of $\X(t)$ converges for every trial to a value less than $\epsilon$, defined $\epsilon = \max_{(u,v) \in \edges} [\W]_{u,v}$.
By Theorem \ref{thm:existence}, this is expected when $\alpha(t) \to 0$. It is interesting that, at least for this set of trials, the loss $\ell(\X(t))$ also converges when $\alpha(t)$ converges to an $\alpha>0$, and the value to which it converges is also less than $\epsilon$. This means that it is possible that a negotiation continues indefinitely, even if the effective equilibrium condition has been satisfied.

To further investigate what happens in these situations, we focus on one of the four trajectories for which $\alpha(t)$ converges to some $\alpha>0$ and take a closer look at the evolution of the value vectors $\X_u(t)$ of each agent; see network in Fig.~\ref{fig:convergence}. We observe that the value vectors of all but three agents converge; the value vectors of the remaining agents decrease (point-wise) to negative-infinity. The non-converging agents $u \in \{4,9,14\}$ form a connected subgraph of $\graph$.

\subsection{Scalability}
In Algorithm \ref{alg:sync}, agents update their values in parallel. We investigated the relationship between the execution time of a single update of the heat equation \eqref{eq:heat} (i.e. Algorithm \ref{alg:sync}, Lines 3-13), and the number of agents, as well as the number of alternatives. By inspection of our results (see Fig.~\ref{fig:scalability}), we believe the time complexity is linear in the number of agents, but further analysis is required to confirm a potential computational advantage of our algorithm over centralized constraint satisfaction algorithms, e.g.~model-checking \cite{baier2008}. We remark that we are able to call the tropical Tarski Laplacian function in just over 8 minutes for a multi-agent system with around 800 agents who are evaluating 20 alternatives.

\begin{figure}
\centering
\includegraphics[width=0.5\textwidth]{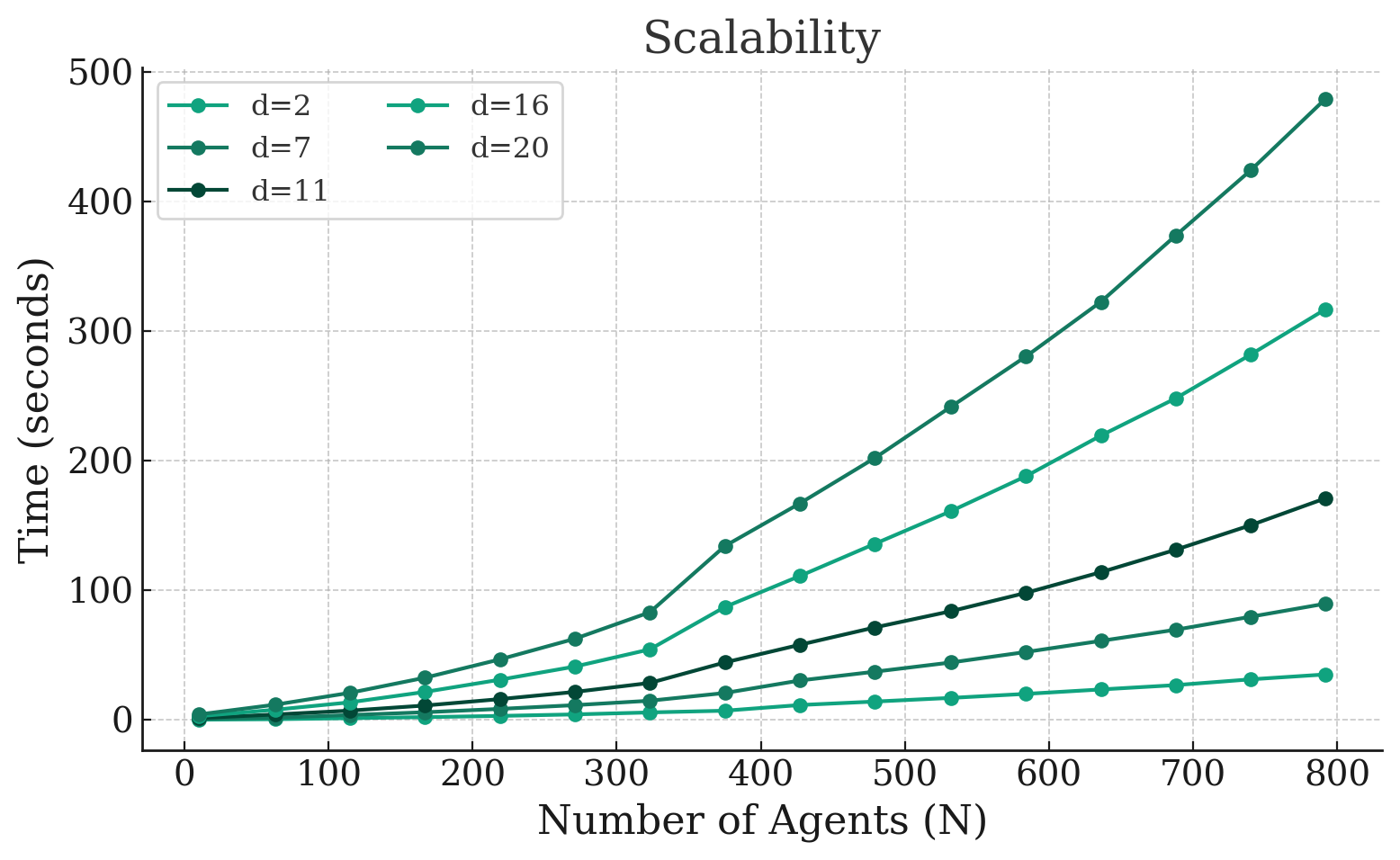}
\caption{Execution time of tropical Tarski Laplacian in the number of agents, $N$, and the number of alternatives, $d$.}
\label{fig:scalability}
\end{figure}

\bibliographystyle{ieeetr}
\bibliography{IEEEabrv,bibliography}

\appendix

    

\begin{proof}[Proof of Proposition \ref{prop:topical}]
    We first recall a lemma (see references \cite{crandall1980,gunawardena1995}) relating the three axioms: monotonicity, homogeneity, and non-expansiveness.
    \begin{lemma}[{\cite[Proposition 2]{crandall1980}, \cite[Proposition 1.1]{gunawardena1995}}]
        Suppose $\F: \R^n \to \R^n$ satisfies monotonicity and additive homogeneity. Then,
        \begin{align*}
            \| \F(\x) - \F(\y) \|_{\infty} \leq \| \x - \y\|_{\infty}
        \end{align*}
        for all $\x, \y \in \R^n$.
    \end{lemma}
    We argue that $\F$ satisfies homogeneity and monotonicity. It suffices to show $\Laplacian$ satisfies these two properties. Suppose $\alpha \in \R$. Then, by repeated application of Lemma \ref{lem:nonlinear}-2,
    \begin{align*}
        \Laplacian(\X + \alpha)_u &=& \bigmeetplus_{v \in \neighbors_u} [\W]_{u,v} + \inv{\A}_{u,v} \meettimes \left( \A_{v,u} \jointimes \left( \X_v + \alpha \right) \right) \\
        &=& \alpha + \bigmeetplus_{v \in \neighbors_u} [\W]_{u,v} + \inv{\A}_{u,v} \meettimes \left( \A_{v,u} \jointimes \X_v \right).
    \end{align*}

    For monotonicity, suppose $\X \preceq \Y$. Then, by applying Lemma \ref{lem:nonlinear}-3, it follows $\A_{v,u} \jointimes \X_v \preceq \A_{v,u} \jointimes \Y_v$, which implies $\inv{\A}_{u,v} \meettimes \left( \A_{v,u} \jointimes \X_v \right) \preceq \inv{\A}_{u,v} \meettimes \left(\A_{v,u} \jointimes \Y_v \right)$, and so forth.
\end{proof}

\begin{proof}[Proof of Theorem \ref{thm:convergence}]
    Suppose $\X(t)$ is a trajectory of the heat equation. By Proposition \ref{prop:topical},
    \begin{align*}
        \| \X(t+2) - \X(t+1) \|_{\infty} &\leq& \| \X(t+1) - \X(t) \|_{\infty}.
    \end{align*}
    Hence, $\alpha(t) = \|\X(t+1) - \X(t)\|$ is a monotonically decreasing sequence bounded below, implying $\alpha(t) \to \alpha$ for some $\alpha \geq 0$.
\end{proof}

\begin{proof}[Proof of Theorem \ref{thm:semimodule}]
    Suppose $\X \in \semimod$ and $\alpha \in \Rmax$. Then, by the proof of Proposition \ref{prop:topical},
    \begin{align*}
        \Laplacian(\X + \alpha) = \Laplacian(\X) + \alpha \geq \X + \alpha
    \end{align*}
    Therefore, $\X + \alpha \in \semimod$. Suppose $\X, \Y \in \semimod$. Then, by Lemma \ref{lem:nonlinear}-1, $\Laplacian(\X \join \Y)_u=$
    \begin{align*}
        \bigmeet_{v \in \neighbors_u} [\W]_{u,v} +' \inv{\A_{u,v}} \meettimes \left( \A_{v,u} \jointimes \X_v \join \A_{v,u} \jointimes \Y_v \right). 
    \end{align*}
    By assumption,
    \begin{align*}
        \bigmeetplus_{v \in \neighbors_u} [\W]_{u,v} +' \inv{\A}_{u,v} \meettimes \left( \A_{v,u} \jointimes \X_v \right) \succeq \X_u, \\
          \bigmeetplus_{v \in \neighbors_u} [\W]_{u,v} +' \inv{\A}_{u,v} \meettimes \left( \A_{v,u} \jointimes \Y_v \right) \succeq \Y_u,
    \end{align*}
    which implies (Lemma \ref{lem:adjoint})
    \begin{align*}
        \A_{v,u} \jointimes \X_v  &\succeq& \A_{u,v} \jointimes \left( \X_u - [\W]_{u,v} \right), \\
        \A_{u,v} \jointimes \X_u  &\succeq& \A_{v,u} \jointimes \left( \X_v - [\W]_{v,u} \right). 
   \end{align*}
   Hence (preceding argument, monotonicity of $\Laplacian$, Lemma \ref{lem:nonlinear}-1),
   \footnotesize
   \begin{align*}
       \bigmeet_{v \in \neighbors_u} [\W]_{u,v} +' \inv{\A}_{u,v} \meettimes \left( \A_{v,u} \jointimes \X_v \join \A_{v,u} \jointimes \Y_v \right) \succeq   \\
       \bigmeet_{v \in \neighbors_u}  [\W]_{u,v} +' \inv{\A}_{u,v} \meettimes \Bigl( \A_{u,v} \jointimes \bigl( \X_u - [\W]_{u,v} \bigr) \Bigr. \\ 
       \Bigl. \quad \join \quad \A_{u,v} \jointimes \Bigl( \Y_u - [\W]_{u,v} \bigr) \Bigr)  \\
       = \bigmeet_{v \in \neighbors_u} \inv{\A}_{u,v} \meettimes \left( \A_{u,v} \jointimes \left( \X_u \join \Y_u \right) \right).
   \end{align*}
   \normalsize
   By Lemma \ref{lem:adjoint}, $\inv{\A}_{u,v} \meettimes \left( \A_{u,v} \jointimes \left( \X_u \join \Y_u \right) \right) \succeq \X_u \join \Y_u$. Hence,
\begin{align*}
        \Laplacian(\X \join \Y)_u  &\succeq& \bigmeet_{v \in \neighbors_u} \inv{\A}_{u,v} \meettimes \left( \A_{u,v} \jointimes \left( \X_u \join \Y_u \right) \right) \\
        &\succeq& \X_u \join \Y_u.
\end{align*}
Therefore, $\X \join \Y \in \semimod$.
\end{proof}

\begin{proof}[Proof of Theorem \ref{thm:existence}]
    $\X \in \semimod$ is equivalent to 
    \begin{align*}
        \bigmeetplus_{v \in \neighbors_u} [\W]_{u,v} +' \inv{\A}_{u,v} \meettimes \left( \A_{v,u} \jointimes \X_v \right) \succeq \X_u \quad \forall u \in \nodes.
    \end{align*}
   It follows by the greatest lower bound property
   \begin{align}
       [\W]_{u,v} +' \inv{\A}_{u,v} \meettimes \left( \A_{v,u} \jointimes \X_v \right) \succeq \X_v \label{eq:thmmain2}
   \end{align}
   holds for all $(u,v)$. If $(u,v) \notin \edges$, then $[\W]_{u,v} = \infty$ and \eqref{eq:thmmain2} automatically holds. If $(u,v) \in \edges$, then, by symmetry,
   \begin{align*}
       [\W]_{u,v} +' \inv{\A}_{u,v} \meettimes \left( \A_{v,u} \jointimes \X_v \right) \succeq \X_u, \\
       [\W]_{v,u} +' \inv{\A}_{v,u} \meettimes \left( \A_{u,v} \jointimes \X_u \right) \succeq \X_v, 
   \end{align*}
   for all $(u,v) \in \edges$. Hence,
    \begin{align*}
       \inv{\A}_{u,v} \meettimes \left( \A_{v,u} \jointimes \X_v \right) \succeq \X_u - [\W]_{u,v},  \\
       \inv{\A}_{v,u} \meettimes \left( \A_{u,v} \jointimes \X_u \right) \succeq \X_v - [\W]_{v,u}.
   \end{align*}
   By Lemma \ref{lem:adjoint},
    \begin{align*}
        \A_{v,u} \jointimes \X_v  \succeq \A_{u,v} \jointimes \left( \X_u - [\W]_{u,v} \right), \\
        \A_{u,v} \jointimes \X_u  \succeq \A_{v,u} \jointimes \left( \X_v - [\W]_{v,u} \right). 
   \end{align*}
   Rearranging terms and Lemma \ref{lem:nonlinear}-2 implies
   \begin{align*}
   -[\W]_{u,v} \preceq \A_{v,u} \jointimes \X_v - \A_{u,v} \jointimes \X_u \preceq [\W]_{v,u}.
   \end{align*}
   By Assumption \ref{assume:scalar},
    \begin{align*}
   \big| \A_{v,u} \jointimes \X_v  - \A_{u,v} \jointimes \X_u \big| \leq [\W]_{u,v} \quad  \forall (u,v) \in \edges,
    \end{align*}
    implying 
    \begin{align*}
        \| \A_{v,u} \jointimes \X_v - \A_{u,v} \jointimes \X_u  \|_{\infty} \leq [\W]_{u,v} \quad \forall (u,v) \in \edges.
    \end{align*}
\end{proof}

\end{document}